\documentclass[lettersize,journal]{IEEEtran}
\usepackage{amsmath,amsfonts}
\usepackage{array}
\usepackage[caption=false,font=normalsize,labelfont=sf,textfont=sf]{subfig}
\usepackage{textcomp}
\usepackage{stfloats}
\usepackage{url}
\usepackage{verbatim}
\usepackage{graphicx}
\def\BibTeX{{\rm B\kern-.05em{\sc i\kern-.025em b}\kern-.08em
    T\kern-.1667em\lower.7ex\hbox{E}\kern-.125emX}}
\usepackage{balance}

\usepackage{graphics} % for pdf, bitmapped graphics files

\usepackage{multirow,multicol}
\usepackage{cite}
\usepackage{booktabs}
\usepackage[colorlinks,linkcolor=blue,hyperindex,CJKbookmarks,pdftex]{hyperref}
\usepackage{array}
\usepackage{color}

\usepackage{algorithm}
\usepackage{algpseudocode}
\usepackage{amssymb}
\newcommand{\col}{\mbox{col}}
\usepackage{amsthm}
\newtheorem{thm}{Theorem}[section]

\newtheorem{prop}[thm]{Proposition}

\theoremstyle{definition}

\begin{document}

\title{Learning-Based Tracking Perimeter Control for Two-region Macroscopic Traffic Dynamics}

\author{
Can~Chen, Yunping~Huang, Hongwei~Zhang, Shimin~Wang, Martin~Guay, Shu-Chien~Hsu, Renxin~Zhong
\thanks{This work was supported by the National Natural Science Foundation of China under Grant 72071214 and the Research Grants Council of Hong Kong under Grant 15211518E. (Corresponding authors: Shu-Chien~Hsu; Renxin~Zhong.)}
\thanks{Can~Chen, and Renxin~Zhong are with School of Intelligent Systems Engineering, Sun Yat-Sen University, Guangzhou, 510000, China (e-mail: can-caesar.chen@connect.polyu.hk; zhrenxin@mail.sysu.edu.cn). }
\thanks{Yunping~Huang is with Department of Civil and Environmental Engineering, Schulich School of Engineering, University of Calgary, Calgary, Alberta, AB T2L 0Y2, Canada (e-mail: yunping.huang@connect.polyu.hk).}
\thanks{Hongwei~Zhang is with School of Intelligence Science and Engineering, Harbin Institute of Technology, Shenzhen, 518055, China (e-mail: hwzhang@hit.edu.cn).}
\thanks{Shimin~Wang is with Massachusetts Institute of Technology, Cambridge, MA 02139, USA (e-mail: bellewsm@mit.edu).}
\thanks{Martin~Guay is with Smith Engineering, Queen's University, Kingston, ON K7L 3N6, Canada (e-mail: guaym@queensu.ca).}
\thanks{Shu-Chien~Hsu is with Department of Civil and Environmental Engineering, The Hong Kong Polytechnic University, Hong Kong, 999077, China (e-mail: mark.hsu@polyu.edu.hk).}
}

% \markboth{IEEE Transactions on Neural Networks and Learning Systems,~Vol.~XX, No.~XX, May~2025}%
\markboth{Preprint submitted to arXiv, May~2025}%
{Chen \MakeLowercase{\textit{et al.}}: Learning-Based Tracking Perimeter Control for Two-region Macroscopic Traffic Dynamics}

\maketitle

\begin{abstract}
Leveraging the concept of the macroscopic fundamental diagram (MFD), perimeter control can alleviate network-level congestion by identifying critical intersections and regulating them effectively. Considering the time-varying nature of travel demand and the equilibrium of accumulation state, we extend the conventional set-point perimeter control (SPC) problem for the two-region MFD system as an optimal tracking perimeter control problem (OTPCP). Unlike the SPC schemes that stabilize the traffic dynamics to the desired equilibrium point, the proposed tracking perimeter control (TPC) scheme regulates the traffic dynamics to a desired trajectory in a differential framework. Due to the inherent network uncertainties, such as heterogeneity of traffic dynamics and demand disturbance, the system dynamics could be uncertain or even unknown. To address these issues, we propose an adaptive dynamic programming (ADP) approach to solving the OTPCP without utilizing the well-calibrated system dynamics. Numerical experiments demonstrate the effectiveness of the proposed ADP-based TPC. Compared with the SPC scheme, the proposed TPC scheme achieves a 20.01\% reduction in total travel time and a 3.15\% improvement in cumulative trip completion. Moreover, the proposed adaptive TPC approach can regulate the accumulation state under network uncertainties and demand disturbances to the desired time-varying equilibrium trajectory that aims to maximize the trip completion under a nominal demand pattern. These results validate the robustness of the adaptive TPC approach.
\end{abstract}


%%%%%%%%%%%%%%%%%%%%%%%%%%%%%%%%%%%%%%%%%%%%%%%%%%%%%%%%%%%%%%%%%%%%%%%%%%%%%%%%
\section{Introduction}

\IEEEPARstart{R}{apid} urbanization has led to a significant increase in car usage across metropolitan cities worldwide, resulting in a surge in traffic congestion, accidents, and pollution. Managing the rapidly growing travel demand requires the efficient use of existing infrastructure through appropriate traffic control schemes. The conventional traffic control methods, such as local traffic signal control strategies, focus on link-level strategies, which may not be optimal or might not achieve the stabilization of the system for heterogeneous networks with multiple bottlenecks and heavily directional demand flows. Network-level traffic control strategies should be developed to alleviate network congestion when confronted with conditions.

The adoption of the macroscopic fundamental diagrams (MFDs) to model and regulate the traffic flow of large-scale urban networks has been extensively studied in the last decade (see \cite{huang2024comparison} and the references therein).
The MFD intuitively provides an aggregate, low-scatter relationship between the network vehicle density (veh/km) (or accumulation (veh)) and network outflow (or trip completion flow rate (veh/h)).
Under the MFD framework, a heterogeneous urban traffic network is divided into several homogeneous regions, with each admitting a well-defined MFD \cite{ji2012spatial}. Such an analytically simple and computationally efficient framework enables MFD to be a promising solution for large-scale network control and management.

The perimeter control is one of the most significant applications of MFDs. The perimeter control aims to manipulate the transfer flow at the boundaries of the region, which is a promising solution to alleviating network-scale traffic congestion. Considerable research efforts have been dedicated to devising optimal network traffic control strategies based on MFDs. Apart from previous literature on maximizing the network traffic throughput by leveraging perimeter control \cite{daganzo2007urban,geroliminis2013optimal,aalipour2018analytical}, considerable research efforts have been devoted to devising perimeter control strategies that regulate the network accumulation to the desired equilibrium, i.e., set-point perimeter control \cite{aboudolas2013perimeter,keyvan2015multiple,haddad2016adaptive}. The robust perimeter control problem of the MFD-based system was also addressed in previous studies, e.g., \cite{haddad2015robust,zhong2018robust}.

A critical assumption adopted in traffic control (including the perimeter control and signal control) is that the steady state of the system can be achieved, and the equilibria of the system can be determined. Under this assumption, the stability of fixed equilibrium points in the sense of Lyapunov is widely applied in traffic control. Considering the dynamic nature of traffic demand and supply, especially for fast time-varying cases, identification of the steady state is an extremely difficult and unclear task in practice \cite{zhong2018robust,ZHONG2018327}. Some recent studies have attempted to optimize set-points for traffic control by updating them based on real-time traffic state estimations/measurements \cite{wang2021feedback,mohajerpoor2020h}, and others by using data-driven approaches \cite{kouvelas2017enhancing} or Nash equilibrium seeking schemes \cite{kutadinata2016enhancing}.
Model predictive control (MPC) was utilized by \cite{yu2020two} to optimize the set-point of MFD dynamics and used this set-point with an iterative learning method to design traffic signal timing plans. They conducted the planning and control synchronously, which may cause the curse of dimensionality in large-scale urban networks with many intersections that are distributedly managed. However, there is a complex and unclear relationship between the network traffic performance and desired set point, with no unified or clear definition of the best set point in an ever-changing environment.

The aforementioned studies on perimeter control can be regarded as model-based traffic-responsive control, which assumes that model parameters are accurately calibrated and perfect knowledge of the network is available. Note that traffic networks are subject to various uncertainties (e.g., demand noise and model error), making these assumptions difficult and even impossible to be met. Recently, model-free methods, such as iterative learning control \cite{REN2020102618}, model-free adaptive predictive control \cite{li2022distributed}, deep reinforcement learning (RL) \cite{zhou2021model}, and adaptive dynamic programming (ADP) \cite{SU2020102628,chen2022data,chen2024iterative} have been proposed to address the problem of devising adaptive perimeter control strategies for MFD systems with unknown system dynamics. The RL and ADP bridge the gap between optimal control and adaptive control. In an off-line manner, the ADP method provides an approximate solution to the optimal control problem obtained from the Bellman optimality principle and the dynamic programming principle (i.e., the Hamilton-Jacobi-Bellman (HJB) equation).

This paper proposes a model-free ADP-based tracking perimeter control (TPC) scheme that extends the stability of a single equilibrium (or its invariant set) to a desired trajectory. Few existing studies have addressed the TPC problem for MFD-based traffic networks except \cite{haddad2017coordinated}. Based on this work, \cite{HADDAD2020133} investigated the effect of constant time delays. Local linearization around the desired equilibrium was performed in both works to simplify the controller design. Different from the existing works, we explicitly consider both state and control constraints in the solution of the TPC problem and no model linearization is required. In this paper, first, the optimal tracking perimeter control problem (OTPCP) is transformed into the minimization of a nonquadratic performance function subject to an augmented system composed of the original system and the command generator system. Then, an ADP algorithm is proposed to generate the optimal solution to the associated HJB equation without complete knowledge of the augmented dynamics.

The remainder of this paper is organized as follows: \autoref{sec:pf} proposes a novel problem formulation of the OTPCP for the two-region MFD system. \autoref{sec:apc} develops a model-free ADP approach to solving the OTPCP. Then experimental results are presented in \autoref{sec:exp}. Finally, \autoref{sec:con} concludes this paper.

%%%%%%%%%%%%%%%%%%%%%%%%%%%%%%%%%%%%%%%%%%%%%%%%%%%%%%%%%%%%%%%%%%%%%%%%%%%%%%%%
\section{Optimal tracking perimeter control of a two-region MFD system}\label{sec:pf}

\subsection{Two-region MFD system dynamics}

\begin{figure}[htbp]
  \centering
  \includegraphics[width=3.1in]{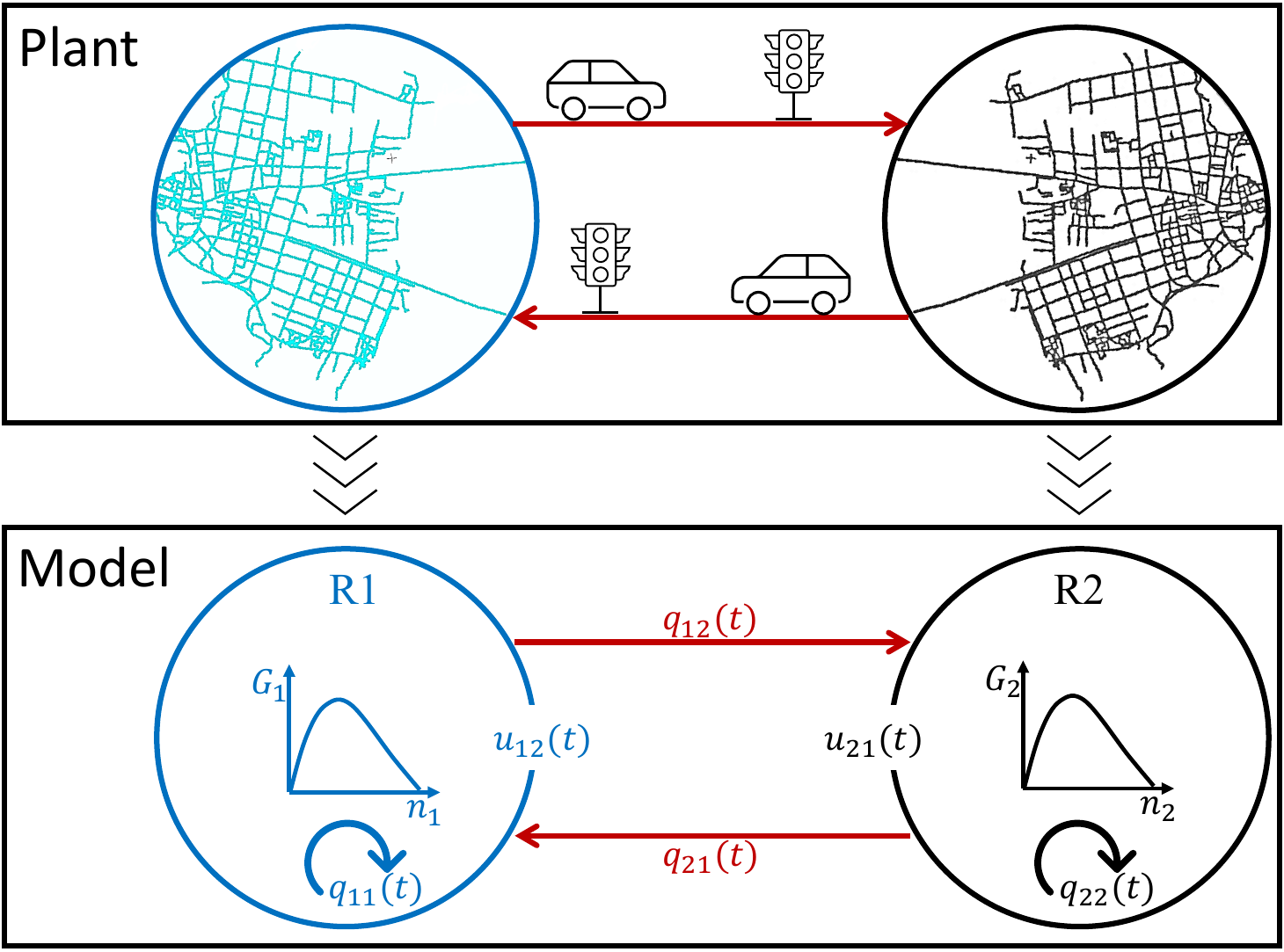}\\
  \caption{The two-region MFD system}\label{fig:mfdtopo}
\end{figure}

To begin with, we recapitulate the two-region MFD system model. An urban network with two regions shown in \autoref{fig:mfdtopo} is of great significance in investigating gating control on the periphery. Assume that the urban network is composed of two homogeneous regions that both admit well-defined MFDs. A two-region MFD system can be used to model the macroscopic traffic dynamics. The MFD is a function that depicts a nonlinear relationship between the regional accumulation $n_i(t)$ (veh) and the trip completion rate $G_i(n_i(t))$ (veh/s) at time $t$, $i=1,\ 2$. The regional accumulation $n_i(t)$ represents the number of vehicles in region $i$ with $0\leq n_i \leq n^{jam}_i$ where $n^{jam}_i$ is the jam accumulation, i.e., the maximum vehicle number in region $i$. Let $q_{ij} (t)$ (veh/s) denote the travel demand generated in region $i$ with destination to region $j$. By distinguishing whether the origin and destination of the travel demand are in the same region or not, the travel demand can be divided into endogenous and exogenous travel demand. Corresponding to the travel demand, four state variables, denoted by $n_{ij} (t)$ (veh) are identified. These state variables represent the accumulations contributed by the travel demand from region $i$ to region $j$. By definition, we have $n_i (t) = \sum_{j} n_{ij} (t)$. Meanwhile, the perimeter control variables are introduced to the system, denoted as $u_{12} (t)$ and $u_{21} (t)$ with $0\leq u_{ij} (t)\leq 1,\ i \neq j$, which are utilized to control the transfer flow between R1 and R2 on the border.
The perimeter control, essentially, is a type of gating control actualized on the boundaries, which can be implemented using the existing infrastructure such as traffic signals \cite{chen2022data,su2023hierarchical,tsitsokas2023two}.
The sending function of the transfer flow from region $i$ with destination region $j$ can be calculated by $\frac{n_{ij}(t)}{n_i(t)}G_i(n_i(t))$, and the completed transfer flow is determined by the perimeter control, i.e., $u_{ij} (t)\frac{n_{ij}(t)}{n_i(t)}G_i(n_i(t))$. On the other hand, the completed internal flow is defined as $\frac{n_{ii}(t)}{n_i(t)}G_i(n_i(t))$.

Based on flow conservation, the dynamics of the two-region MFD system can be regarded as a class of non-affine systems
\begin{equation}\label{eq:onadyn}
    \dot{n}(t)=K(n(t),u(t),q(t))
\end{equation}
where $n(t)=\col(n_{11}(t),n_{12}(t),n_{21}(t),n_{22}(t))\in \mathbb{R}^4_+$, $u(t)=\col(u_{12}(t),u_{21}(t))\in \mathbb{R}^2_+$, and $q(t)=\col(q_{11}(t),q_{12}(t),q_{21}(t),q_{22}(t))\in \mathbb{R}^4_+$ are the accumulation state of the system, the perimeter control, and the travel demand, respectively. Here $K(n,u,q)$ has the following well-known form \cite{geroliminis2013optimal}:
\begin{align}\label{eqK}
  K(n,u,q)\triangleq \left[
  \begin{array}{c}
   -\frac{n_{11}}{n_1}G_1(n_1)+\frac{n_{21}}{n_2}G_2(n_2)u_{21}+q_{11} \\
   -\frac{n_{12}}{n_1}G_1(n_1)u_{12}+q_{12} \\
   -\frac{n_{21}}{n_2}G_2(n_2)u_{21}+q_{21} \\
   -\frac{n_{22}}{n_2}G_2(n_2)+\frac{n_{12}}{n_1}G_1(n_1)u_{12}+q_{22}
  \end{array} \right]
\end{align}
subject to
\begin{equation*}%\label{oscuc}
  0\leq n_{i}(t)\leq n^{jam}_{i},\ 0\leq u^{\min}_{ij}\leq u_{ij}(t)\leq u^{\max}_{ij}\leq 1
\end{equation*}

\subsection{Optimal TPC problem formulation}

For the TPC problem, the perimeter control is designed to manipulate the cross-boundary flows such that the accumulation state $n$ can track a desired trajectory. For different periods of within-day traffic (i.e., off-peak period and peak period), the desired control targets should be different and fit the dynamics of the demand pattern. As shown in \autoref{fig:consch}, during the first off-peak period (e.g., 0:00--7:00), no control is necessary as the travel demand is low. Regulation is then claimed at the onset of the congestion. During the peak period (e.g., 7:00--11:00) with high travel demand, a steady accumulation state (target equilibrium) less than but close to the critical accumulation of the MFD is desired because operating the protected region around the critical accumulation maximizes its throughput \cite{ZHONG2018327}. After the peak-period congestion dissolves, the second off-peak period (e.g., 11:00--16:00) commences, during which the demand level is medium. It is not necessary to set the target equilibrium to be close to the critical accumulation that is larger than the steady state yielded by the demand pattern. Hence, the change of the control target is desired for the second off-peak period.

It is desired to design a reference signal associated with the performance of the network traffic flows. Moreover, a reference signal, if bounded by the desired invariant set of the steady states, would be more desirable because that means the control target is more achievable, controllable, and practical.

\begin{figure}[!h]
  \centering
  \includegraphics[width=3.3in]{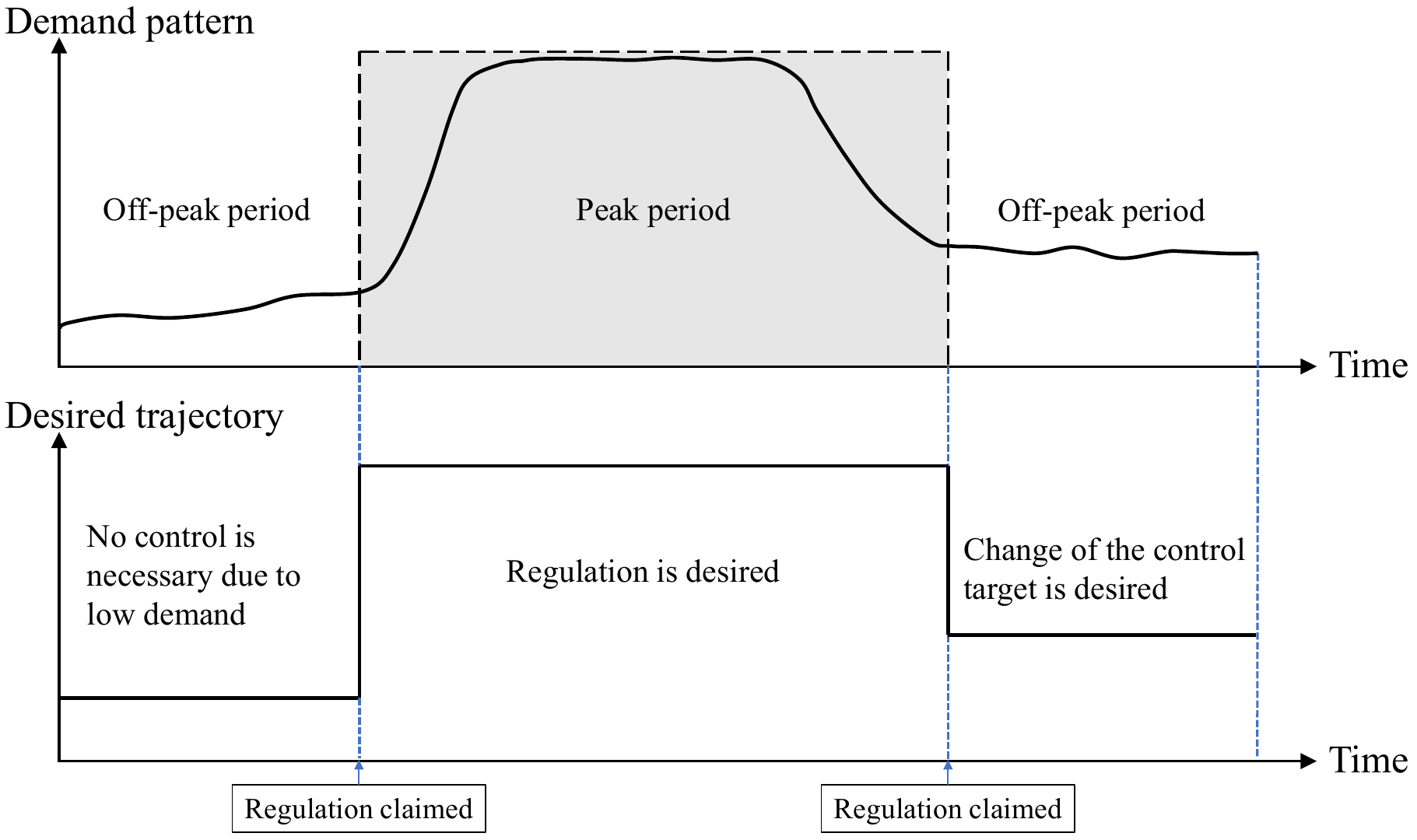}\\
  \caption{Demand pattern and desired state trajectory}\label{fig:consch}
\end{figure}

Now we present the formulation of OTPCP. Following \cite{zhong2018robust}, let $n=\col(n_{11}, n_{12}, n_{21}, n_{22})$ and $u=\col(u_{12}, u_{21})$ denote the state vector and control input vector respectively, then the two-region MFD traffic dynamics can be written in the following affine form:
\begin{equation}\label{oafsys}
  \dot{{n}}=f({n})+s({n})\cdot{u}
\end{equation}
where $f({n})$ and $s({n})$ are the drift dynamics and input dynamics given by
\begin{equation*}
    \begin{split}
        & f(n)=\left[
        \begin{array}{c}
          -\frac{n_{11}}{n_1}G_1(n_1)+q_{11} \\
          q_{12} \\
          q_{21} \\
          -\frac{n_{22}}{n_2}G_2(n_2)+q_{22}
        \end{array}
          \right] \\
        & s(n)=\left[
        \begin{array}{cc}
          0 & \frac{n_{21}}{n_2}G_2(n_2) \\
          -\frac{n_{12}}{n_1}G_1(n_1) & 0\\
          0 & -\frac{n_{21}}{n_2}G_2(n_2) \\
          \frac{n_{12}}{n_1}G_1(n_1) & 0
        \end{array}
          \right]
    \end{split}
\end{equation*}

We denote the reference trajectory by ${n}^d(t)=\col({n}^{d}_{11}(t),{n}^{d}_{12}(t),{n}^{d}_{21}(t),{n}^{d}_{22}(t))$. The trajectory ${n}^d(t)$ is assumed to be bounded and could be generated by the following Lipschitz continuous command generator dynamics \cite{zhang2017tracking}:
\begin{equation}\label{cgd}
  \dot{{n}}^d = \theta({n}^d(t))
\end{equation}
and $\theta(0)=0$.

The target of the OTPCP is to find an optimal controller ${u}^*(t)$ to make the state ${n}(t)$ can track the desired state ${n}^d(t)$. We define $e^d(t)={n}(t)-{n}^d(t)$ as the tracking error, and the tracking error dynamics are given by
\begin{equation}\label{errdyn}
\begin{split}
  \dot{e}^d & =\dot{{n}}-\dot{{n}}^d =f({n})+s({n})\cdot{u}-\theta({n}^d) \\
  & = K(e^d+{n}^d,{u},q)-\theta({n}^d)
\end{split}
\end{equation}

As reported by \cite{zhang2017tracking}, the standard solution to the optimal tracking control problem is composed of two parts: 1) the steady-state part of the control input ${u}_s(t)$ that guarantees perfect tracking of the reference trajectory, and 2) the feedback part of the control input $\mu(t)$ that stabilizes the tracking error dynamics in an optimal manner, i.e.,
\begin{equation*}
  {u}(t)={u}_s(t)+\mu(t)
\end{equation*}

Suppose that perfect information on the dynamics is available and the inverse of the input dynamics $s^{-1}({n}^d)$ exists, the steady-state part control input $\tilde{u}_s(t)$ can be obtained by
\begin{equation}\label{ssucom}
{u}_s=s^{-1}({n}^d)(\theta({n}^d)-f({n}^d))
\end{equation}
On the other hand, as reported by \cite{zhang2018near}, if $s^{-1}({n}^d)$ does not exist, ${u}_s(t)$ can be developed by
\begin{equation}\label{ssucomi}
  {u}_s=[s({n}^d)^Ts({n}^d)]^{-1}s({n}^d)^T[\theta({n}^d)-f({n}^d)]
\end{equation}

The feedback part of the control $\mu(t)$ can be obtained by minimizing the following performance function:
\begin{equation}\label{OOtc}
  V(t)=\int^{\infty}_{t} (e^d(\tau))^TQe^d(\tau) + U(\mu(\tau)) \mathrm{d}\tau
\end{equation}
where $Q$ is a symmetric positive definite matrix of proper dimension and $U(\mu)$ is a positive definite function that computes the control effort.

In this paper, we propose an augmented system for the OTPCP. Hence the OTPCP can be transformed into the minimization of a nonquadratic performance function subject to an augmented system composed of the original system and the command generator system. We define the augmented system state $N(t)=[(e^d(t))^T,\ ({n}^d(t))^T]^T$ and the augmented system as
\begin{equation}\label{augsys}
  \dot{N}(t)=F(N(t))+S(N(t))\times\mu(t)
\end{equation}
where
\begin{equation*}%\label{FS}
\begin{split}
  & F(N)=\left[
  \begin{array}{c}
  f(e^d+{n}^d)+s(e^d+{n}^d){u}_s-\theta({n}^d) \\
  \theta({n}^d)
  \end{array}
  \right] \\
  & S(N)=\left[
  \begin{array}{c}
  s(e^d+{n}^d) \\
  0
  \end{array}
  \right]
\end{split}
\end{equation*}

For this augmented system, we introduce the following performance function.
\begin{equation}\label{npf}
  \bar{V}(N(t))=\int^{\infty}_{t} \big(\bar{Q}(N(\tau))+ \bar{U}(\mu(\tau))\big) \mathrm{d}\tau
\end{equation}
where
\begin{equation*}
\begin{split}
\bar{Q}(N(t)) & = (N(t))^T\hat{Q}N(t) \\
& =\left[
\begin{array}{c}
 e^d  \\
 {n}^d  \\
\end{array}
\right]^T\left[
\begin{array}{cc}
Q & 0 \\
0 & 0 \\
\end{array}
\right]\left[
\begin{array}{c}
e^d  \\
{n}^d  \\
\end{array}
\right]
\end{split}
\end{equation*}
and $\bar{U}(\mu)$ is a positive definite integrand function defined as
\begin{equation}\label{satU}
  \bar{U}(\mu)= 2\int^{\mu}_{0} \left(\lambda\tanh^{-1}\left(\frac{\upsilon}{\lambda}\right)\right)Rd\upsilon
\end{equation}
where $\upsilon\in \mathbb{R}^2$, $\lambda$ is the saturating bound for the actuators and without loss of generality, $R=\textnormal{diag}(\gamma_1,\gamma_2)$ is a positive semidefinite symmetric matrix. This nonquadratic performance function is used in the optimal regulation problem of constrained-input systems to deal with the input constraints. In fact, using this nonquadratic performance function, the following constraints are always satisfied, i.e., $|\mu_i(t)|\leq\lambda,\ i=1,2$.

By constructing the augmented system \eqref{augsys}, the conventional standard solution to the OTPCP is transformed into solving the optimal feedback part $\mu(N)$ for the augmented system, whereas the solution to the steady-state control $\tilde{u}_s$ has been substituted in the dynamics $F(N)$.

Based on the Bellman optimality principle, suppose $V^*$ is the optimal value function, then it satisfies the following tracking HJB equation
\begin{equation}\label{eq:iniTHJB}
\begin{split}
  H(N,\mu,V^*)= & N^T\hat{Q}N + 2\int^{\mu}_{0} \left(\lambda\tanh^{-1}\left(\frac{\upsilon}{\lambda}\right)\right)Rd\upsilon \\
  & + (\nabla V^{*}(N))^T\times (F(N)+S(N)\mu)
\end{split}
\end{equation}
where $\nabla$ is the symbol of partial derivative.

Applying the stationary condition $\partial H/\partial \mu^*=0$, the optimal control policy is given by
\begin{equation}\label{aocp}
\begin{split}
  \mu^*(N) & =\arg\min_{\mu\in\Lambda(\Omega)} H(N,\mu,V^*) \\
  & =-\lambda\tanh\left(\frac{1}{2\lambda}R^{-1}S(N)^T\nabla V^*(N)\right)
\end{split}
\end{equation}
Substituting \eqref{aocp} into \eqref{satU} results in
\begin{equation}\label{satUs}
\begin{split}
  \bar{U}(\mu^*)= &\ \lambda\nabla V^{*T}(N)S(N)\tanh(D^*) \\
  & +\lambda^2\bar{R}\ln(\underline{\mathbf{1}}-\tanh^2(D^*))
\end{split}
\end{equation}
where $D^*=(1/2\lambda)R^{-1}S(N)^T\nabla V^*(N)$, $\underline{\mathbf{1}}$ is a column vector with all elements being ones and $\bar{R}=[\gamma_1,\gamma_2]\in \mathbb{R}^{1\times 2}$.

Hence, combining \eqref{eq:iniTHJB}-\eqref{satUs}, solving the OTPCP of the two-region macroscopic traffic dynamics is equivalent to solving the optimal value function $V^*$ and policy function $\mu^*$ (i.e., $D^*$) from the following tracking HJB equation.
\begin{equation}\label{finHJB}
\begin{split}
  H(N,\mu^*,\nabla V^*)= &\ N^T\hat{Q}N + \nabla V^{*T}(N)F(N) \\
  & + \lambda^2\bar{R}\ln(\underline{\mathbf{1}}-\tanh^2(D^*))=0
\end{split}
\end{equation}

%%%%%%%%%%%%%%%%%%%%%%%%%%%%%%%%%%%%%%%%%%%%%%%%%%%%%%%%%%%%%%%%%%%%%%%%%%%%%%%%
\section{Adaptive optimal tracking perimeter controller design}
\label{sec:apc}

Due to the strong nonlinearity of the tracking HJB equation \eqref{finHJB}, it is extremely difficult to obtain the analytical solution to \eqref{finHJB}. The offline policy iteration method \cite{lewis2009reinforcement} is one of the most common approaches to resolving this difficulty. First, we revisit the offline policy iteration method, based on which the model-free ADP algorithm is derived, to solve the HJB equation \eqref{finHJB}.
The principle of the offline policy iteration method consists of the following two iterative steps to calculate the Bellman equation and the optimal controller:
\begin{enumerate}
  \item (Policy evaluation) Given an initial admissible control policy $\mu^{(0)}(N)$ and initial cost $V^{(0)}=0$, find $V^{(k)}(N)$ successively approximated by solving the following equation
      \begin{equation}\label{eq8}
      \begin{split}
        & N^T\hat{Q}N+U(\mu^{(k)}(N))+ \left(\nabla V^{(k+1)}(N)\right)^{T} \\
        & \times\left(F(N)+S(N)\mu^{(k)}(N)\right)=0,\ k=0,1,\ldots
      \end{split}
      \end{equation}
  \item (Policy improvement) Update the control policy simultaneously by
      \begin{equation}\label{eq9}
      \begin{split}
         \mu^{(k+1)}(N)&=-\lambda\tanh\left(D^{(k+1)}\right) \\
         D^{(k+1)}&=\frac{1}{2\lambda}R^{-1} S(N)^{T} \nabla V^{(k+1)}(N)
      \end{split}
      \end{equation}
\end{enumerate}
where $k$ is the iterative index.
The convergence of sequence $\{(V^{(k)},\mu^{(k)})\}$ (i.e., $\{(V^{(k)},D^{(k)})\}$) to the optimal value and policy function $(V^*,\mu^*)$ (i.e., $(V^*,D^*)$) applying \eqref{eq8}-\eqref{eq9} has been checked by \cite{chen2022data}.

The urban network is subject to inherent network uncertainties, such as uncertain dynamics of heterogeneity and demand variations.
As shown in \autoref{fig:div_demand}, time-series demand is generated from a normal distribution based on various total inputs $N = \{25000, 30000, 32000\}$ (veh).
Scattered plots of MFDs aggregated at 5-minute intervals under different demand patterns are depicted in \autoref{fig:mfds_dmvar}.\footnote{The data is collected from microscopic simulations, which are carried out in an open-source traffic simulator CityFlow developed by \cite{zhang2019cityflow}. Readers are referred to \cite{ma2025calibration} for more details and simulation results.}
We can observe a discernible clockwise hysteresis in all cases, but with different sizes.
These results indicate the impact of travel demand on MFD hysteresis. Hence, the MFD parameters could be time-varying and uncertain \cite{ma2025calibration}, i.e., $F(N)$ and $S(N)$ could be uncertain and even unknown.

\begin{figure}[!h]
    \centering
    \includegraphics[width=3.3in]{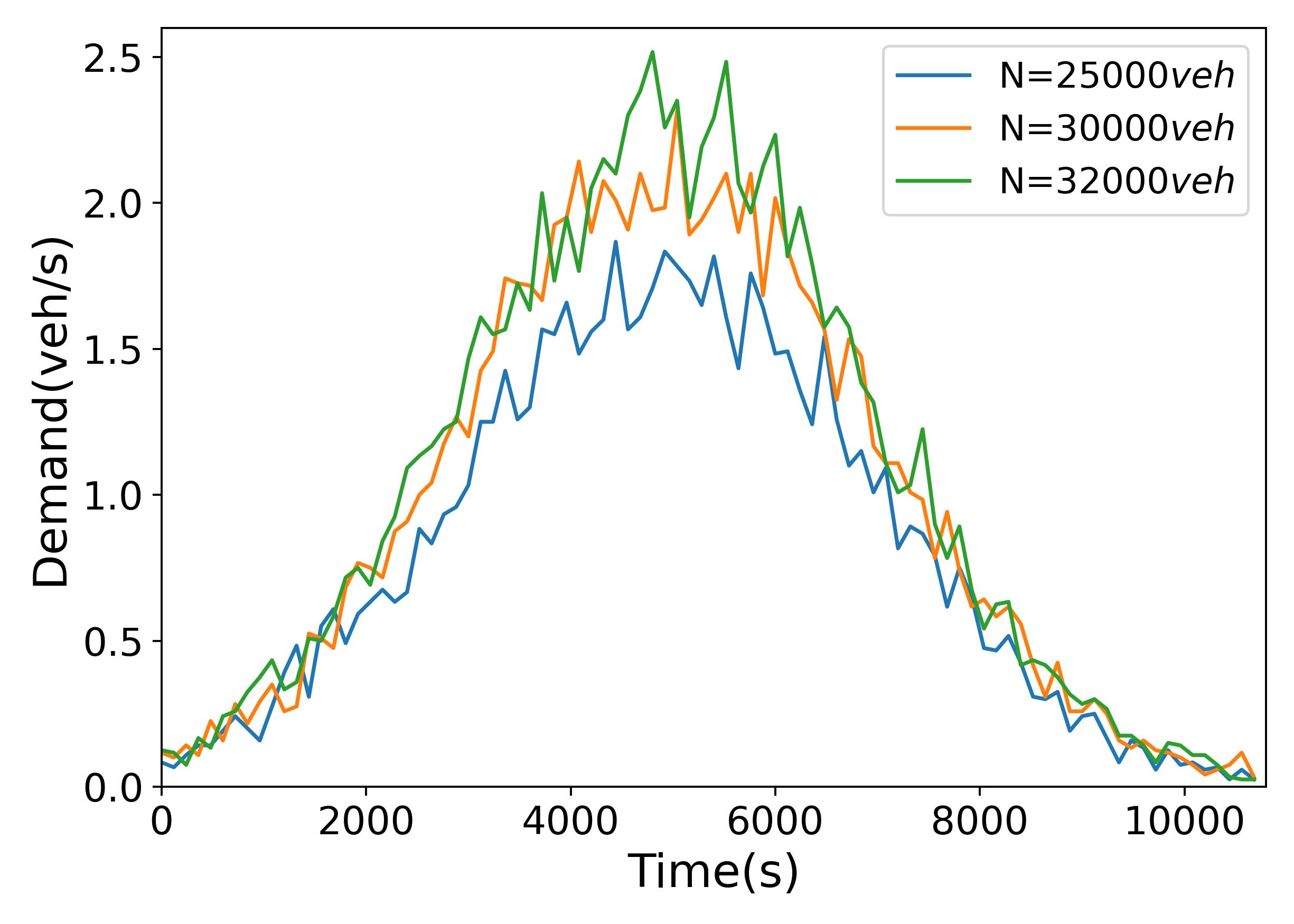}
    \caption{The diverse time-varying demand patterns}
    \label{fig:div_demand}
\end{figure}

\begin{figure}[!h]
    \centering
    \includegraphics[width=3.3in]{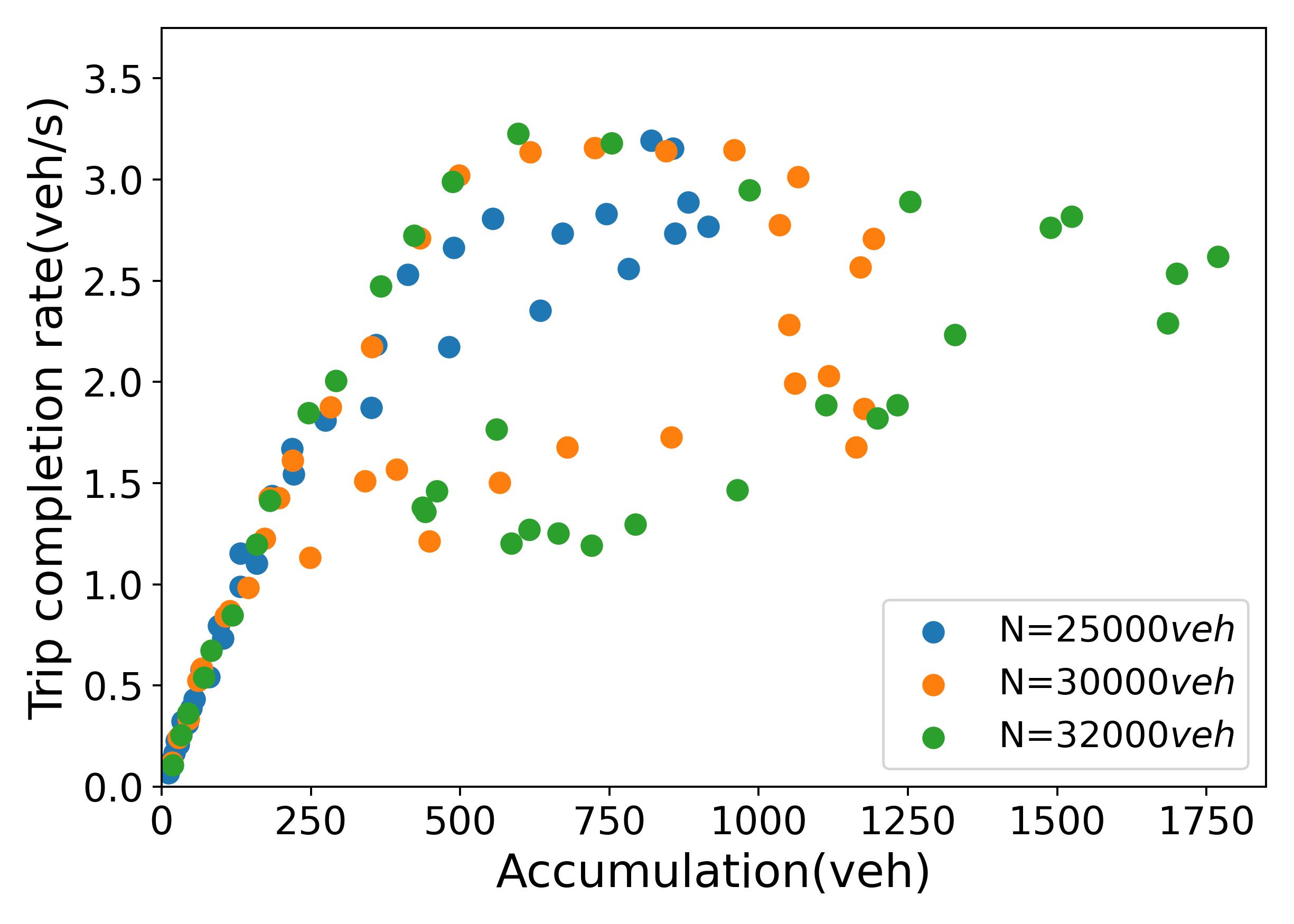}
    \caption{MFDs under demand variations}
    \label{fig:mfds_dmvar}
\end{figure}

Recent advances in ADP-based approaches indicate that model-free methods are compelling for optimizing traffic performance under uncertain and even unknown traffic dynamics \cite{chen2022data,chen2024iterative}.
To implement the model-free method, an improved data-driven algorithm is developed by eliminating the dynamics in the iteration procedure. Denote $\mu^{(k)}$ as the policy to be updated and $\mu$ as the behavior policy that is actually implemented to generate the data for learning. Then we can rewrite the augmented system as:
\begin{equation}\label{opaugsys}
  \dot{N}=F(N)+S(N)\cdot\mu^{(k)}+S(N)\times\left(\mu-\mu^{(k)}\right)
\end{equation}

Taking the derivative of $V^{(k+1)}(N)$ along the system trajectory \eqref{opaugsys} yields
\begin{equation}\label{eq:derv}
\begin{split}
 \frac{\mathrm{d} V^{(k+1)}(N)}{\mathrm{d} t} = & \left(\nabla V^{(k+1)}(N)\right)^{T} \left( F(N)+S(N)\mu^{(k)} \right. \\
 & \left. +S(N)\left(\mu-\mu^{(k)} \right)\right) \\
 = & \left(\nabla V^{(k+1)}(N)\right)^{T} \left( F(N)+S(N)\mu^{(k)} \right) \\
 & + \left(\nabla V^{(k+1)}(N)\right)^{T} S(N)\left(\mu-\mu^{(k)}\right)
\end{split}
\end{equation}

Subtracting \eqref{eq8} from \eqref{eq:derv} and substituting \eqref{eq9} into \eqref{eq:derv}, we have
\begin{equation}\label{itdcost}
\begin{split}
     \frac{\mathrm{d} V^{(k+1)}(N)}{\mathrm{d} t} = & -N^T\hat{Q}N - U(\mu^{(k)}) \\
     & + \left(\nabla V^{(k+1)}\right)^{T} S\left(\mu-\mu^{(k)}\right) \\
     = & -N^T\hat{Q}N - 2\lambda\int^{\mu^{(k)}}_{0} \tanh^{-T}(\frac{\upsilon}{\lambda})R\mathrm{d}\upsilon \\
     & + 2\lambda\tanh^{-T}\left(\frac{\mu^{(k+1)}}{\lambda}\right) R\left(\mu^{(k)}-\mu\right) \\
     = & -N^T\hat{Q}N - 2\lambda\int^{\mu^{(k)}}_{0} \tanh^{-T}(\frac{\upsilon}{\lambda})R\mathrm{d}\upsilon \\
     & - (\nabla V^{(k)})^T S (\mu^{(k)}-\mu) \\
     = & -N^T\hat{Q}N - 2\lambda\int^{\mu^{(k)}}_{0} \tanh^{-T}(\frac{\upsilon}{\lambda})R\mathrm{d}\upsilon \\
     & - 2\lambda (D^{(k+1)})^T R (\mu^{(k)}-\mu)
\end{split}
\end{equation}
Integrating both sides of \eqref{itdcost} over the time interval $[t, t + \Delta t]$, the ADP algorithm for solving the OTPCP is obtained, which is detailed by the so-called integral reinforcement learning (IRL) Bellman equation \eqref{eq10} in \autoref{alg:IRL}.

\begin{algorithm}[H]
    \caption{ADP for optimal tracking perimeter control}
    \label{alg:IRL}
\begin{algorithmic}[1]
  % \noindent \textbf{Input:}
  \Require
  initial admissible control policy $\mu^{(0)}(N)$ (i.e., $D^{(0)}(N)$) and initial cost $V^{(0)}=0$
  % \newline
  % \noindent \textbf{Output:}
  \Ensure $V^{(k)}(N)$ and $\mu^{(k)}(N)$ (i.e., $D^{(k)}(N)$)
  % \newline
  % \noindent
  \State According to the control policy $D^{(k)}$, $D^{(k+1)}$ and $V^{(k+1)}$ can be solved simultaneously as follows:
  \begin{equation}\label{eq10}
  \begin{split}
    & V^{(k+1)}(N(t+\Delta t)) - V^{(k+1)}(N(t)) = \\
    & -\int^{t+\Delta t}_{t} \bigg( N(\tau)^T\hat{Q}N(\tau) \\
    &\ +2\lambda\int^{\mu^{(k)}(\tau)}_{0}(\tanh^{-1}(\upsilon/\lambda))^TR\mathrm{d} \upsilon\bigg)\mathrm{d} \tau \\
    & - 2\lambda\int^{t+\Delta t}_{t}\left(D^{(k+1)}\right)^T R\left(\mu^{(k)}(N(\tau))-\mu(\tau)\right)\mathrm{d}\tau
  \end{split}
  \end{equation}
  % \noindent
  \State On convergence, set $V^{(k+1)}(N)=V^{(k)}(N)$ and the optimal control is $D^*=D^{(k+1)}(N)$ (i.e., $\mu^*=\mu^{(k+1)}(N)$).
\end{algorithmic}
\end{algorithm}

From \eqref{eq10}, we can see that the proposed ADP algorithm does not require any information on the system dynamics. $V^{(k+1)}$ and $\mu^{(k+1)}$ are solved simultaneously using only the data collected from the system.
Proposition \ref{prop:adp_con} proves the equivalence between the policy iterative equations \eqref{eq8}-\eqref{eq9} and the IRL Bellman equation \eqref{eq10}.

\begin{prop}\label{prop:adp_con}
    The model-free ADP algorithm per \autoref{alg:IRL} and the model-based policy iteration method \eqref{eq8}-\eqref{eq9} give an equivalent solution to the OTPCP.
\end{prop}

\begin{proof}
    To prove Proposition \ref{prop:adp_con}, we only need to justify that the IRL Bellman equation \eqref{eq10} gives the same solution to the value function as the Bellman equation \eqref{eq8} and the same updated control policy as \eqref{eq9}.
    The proof is divided into two folds.

    1) \eqref{eq8}-\eqref{eq9} $\Rightarrow$ \eqref{eq10}. Provided that $(V^{(k+1)},\mu^{(k+1)})$ is the solution of the policy iterative equations \eqref{eq8}-\eqref{eq9}, from the derivation of \eqref{eq10}, one can easily deduce that $(V^{(k+1)},\mu^{(k+1)})$ is the solution of \eqref{eq10}.

    2) \eqref{eq10} $\Rightarrow$ \eqref{eq8}-\eqref{eq9}. Provided that $(V^{(k+1)},\mu^{(k+1)})$ is the solution of the IRL Bellman equation \eqref{eq10} and that $D^{(k+1)}=\frac{1}{2\lambda}R^{-1} S^{T} \nabla V^{(k+1)}$.

    Dividing both sides of \eqref{eq10} by $\Delta t$ and taking limits, i.e.,
    \begin{equation*}
    \begin{split}
        & \lim_{\Delta t\rightarrow 0} \frac{V^{(k+1)}(N(t+\Delta t))-V^{(k+1)}(N(t))}{\Delta t} \\
        = & \lim_{\Delta t\rightarrow 0} -\frac{\int^{t+\Delta t}_{t} (N^T\hat{Q}N + 2\lambda \int^{\mu^{(k)}}_{0} \tanh^{-T}(\frac{v}{\lambda})R\mathrm{d}v)\mathrm{d}\tau}{\Delta t} \\
        & - \lim_{\Delta t\rightarrow 0} \frac{2\lambda\int^{t+\Delta t}_{t} (D^{(k+1)})^T R(\mu^{(k)}-\mu)\mathrm{d}\tau}{\Delta t}
    \end{split}
    \end{equation*}
    we have
    \begin{equation}\label{eq:prop1}
    \begin{split}
        \frac{\mathrm{d} V^{(k+1)}}{\mathrm{d} t} = & -N^T\hat{Q}N - 2\lambda \int^{\mu^{(k)}}_{0} \tanh^{-T}(\frac{v}{\lambda})R\mathrm{d}v \\
        & - 2\lambda (D^{(k+1)})^T R(\mu^{(k)}-\mu)
    \end{split}
    \end{equation}

    Substituting $D^{(k+1)}=\frac{1}{2\lambda}R^{-1} S^{T} \nabla V^{(k+1)}$ into \eqref{eq:prop1}, one obtains
    \begin{equation}\label{eq:prop2}
    \begin{split}
        \frac{\mathrm{d} V^{(k+1)}}{\mathrm{d} t} = & -N^T\hat{Q}N - 2\lambda \int^{\mu^{(k)}}_{0} \tanh^{-T}(\frac{v}{\lambda})R\mathrm{d}v \\
        & - (\nabla V^{(k+1)})^T S (\mu^{(k)}-\mu)
    \end{split}
    \end{equation}

    Combining \eqref{opaugsys} and \eqref{eq:prop2}, it follows that
    \begin{equation*}
    \begin{split}
        \frac{\mathrm{d} V^{(k+1)}}{\mathrm{d} t} = & -N^T\hat{Q}N - 2\lambda\int^{\mu^{(k)}}_{0} \tanh^{-T}(\frac{\upsilon}{\lambda})R\mathrm{d}\upsilon \\
        & + (\nabla V^{(k+1)})^T (F + S\mu) \\
        & - (\nabla V^{(k+1)})^T (F + S\mu^{(k)})
    \end{split}
    \end{equation*}
    Then we have
    \begin{equation}\label{eq:prob3}
    \begin{split}
        & -N^T\hat{Q}N - 2\lambda\int^{\mu^{(k)}}_{0} \tanh^{-T}(\frac{\upsilon}{\lambda})R\mathrm{d}\upsilon \\
        & - (\nabla V^{(k+1)})^T (F + S\mu^{(k)}) \\
        = & \frac{\mathrm{d} V^{(k+1)}}{\mathrm{d} t} - (\nabla V^{(k+1)})^T (F + S\mu)
    \end{split}
    \end{equation}

    Note from \eqref{eq:derv} that $\frac{\mathrm{d} V^{(k+1)}}{\mathrm{d} t} = (\nabla V^{(k+1)})^T (F + S\mu)$. Therefore, the right side of \eqref{eq:prob3} equals 0. Hence,
    \begin{equation}\label{eq:prob4}
    \begin{split}
        & -N^T\hat{Q}N - 2\lambda\int^{\mu^{(k)}}_{0} \tanh^{-T}(\frac{\upsilon}{\lambda})R\mathrm{d}\upsilon \\
        & - (\nabla V^{(k+1)})^T (F + S\mu^{(k)}) \\
        = & -N^T\hat{Q}N - U(\mu^{(k)}) -  \left(\nabla V^{(k+1)}\right)^{T}  (F+S\mu^{(k)})\\
        = & 0
    \end{split}
    \end{equation}
    \eqref{eq:prob4} is the same as \eqref{eq8}.

    This completes the proof.
\end{proof}

By iterating $V^{(k)}$ on the IRL Bellman equation and updating the control policy $\mu^{(k)}$, we can approach both the optimal value function $V^*$ and the optimal tracking perimeter control policy $\mu^*$.

%%%%%%%%%%%%%%%%%%%%%%%%%%%%%%%%%%%%%%%%%%%%%%%%%%%%%%%%%%%%%%%%%%%%%%%%%%%%%%%%
\section{Experimental results} \label{sec:exp}

To show the validity of the proposed ADP algorithm for optimal tracking perimeter control of the two-region MFD system, we provide two illustrative examples. In both examples, the MFD functions for the two regions are assumed to be the same, which are in line with those in \cite{haddad2015robust}, i.e.,
\begin{equation}\label{eq:Ghaddad}
  G_{i}(n_{i}) = \frac{1.4877\cdot 10^{-7} n^3_{i} - 2.9815\cdot 10^{-3} n^2_{i} + 15.0912 n_{i}}{3600}
\end{equation}
For both regions, based on \eqref{eq:Ghaddad}, the jam accumulation is $n^{jam}_i=10000$ (veh), the maximum trip completion rate (i.e., the maximum throughput) is $G^{\max}_{i}=6.3$ (veh/s), and the according critical accumulation state is $n^{cr}_i=3392$ (veh). The sample time interval is $60$ s.

\subsection{Example 1: Time-dependent set-point tracking}

In Example 1, we mimic a realistic scenario of peak-hour traffic. For different periods of peak-hour traffic (e.g., congestion onset, stationary congestion, and congestion dissolving), the desired control targets should be different and fit the dynamic nature of the travel demand. For the first hour, the travel demand is at a medium level, i.e., $q(t) = \col(1.2, 1.6, 1.0, 1.4)$ (veh/s) and the set points are $[n^*_1, n^*_2] = [2000, 2000]$ (veh), i.e., around $60\%$ of the critical accumulation states. For the next 2.5 hours with stationary congestion, the travel demand is $q(t) = \col(1.6, 1.6, 1.6, 1.6)$ (veh/s) and the set points are set to be $[n^*_1, n^*_2] = [3000, 3000]$ (veh), i.e., around $90\%$ of the critical accumulation states. After that, as the congestion dissolves, the travel demand decreases to $q(t) = \col(0.9, 0.9, 0.9, 0.9)$ (veh/s), and the set points are modified to a much lower level, i.e., $[n^*_1, n^*_2] = [1500, 1500]$ (veh). The corresponding equilibrium points of the accumulation and the perimeter control input are given in \autoref{tb:eqe1}. The initial OD-specific initial accumulations are $n_{11}(0) = 450$ (veh), $n_{12}(0) = 1050$ (veh), $n_{21} = 1750$ (veh), $n_{22}(0) = 750$ (veh).

\begin{table}[!htb]
  \centering
  \caption{The equilibrium points of Example 1}
  \label{tb:eqe1}
  \begin{tabular}{|c!{\vrule width 1.5pt} l !{\vrule width 1.5pt}l|}
    % \toprule
    \hline\hline
    Time           & $[n^*_{11},n^*_{12},n^*_{21},n^*_{22}]$ & $[u^*_{12},u^*_{21}]$ \\
    % \midrule
    \hline
    0:00-1:00      & $[814.5, 1185.5, 889.3, 1110.7]$ & $[0.50, 0.42]$ \\
    \hline
    1:01-3:30      & $[1538.9, 1461.1, 1461.1, 1538.9]$ & $[0.53, 0.53]$ \\
    \hline
    3:31-5:00      & $[591.6, 908.4, 908.4, 591.6]$ & $[0.33, 0.33]$ \\
    % \bottomrule
    \hline\hline
  \end{tabular}
\end{table}

In Example 1, we compare the performance of the proposed ADP-based TPC against that of the ADP-based SPC. The SPC aims to track a fixed set point $[n^*_1, n^*_2] = [3000, 3000]$ (veh) regardless of the changes in the demand pattern. The results of the accumulation state evolution and OD-specific state evolution of Example 1 are shown in \autoref{fig:state2e2} and \autoref{fig:state4e2}, respectively. The results indicate that the proposed ADP-based TPC scheme can adapt to the changes of the travel demand pattern and regulate the accumulation states (black solid lines) to the corresponding set points (blue dotted lines). Instead of tracking the time-dependent reference trajectory, the SPC scheme succeeds in stabilizing the accumulation states to a fixed set-point in this time-varying demand case (red solid lines).

\begin{figure}[htbp]
  \centering
  \includegraphics[width=3.3in]{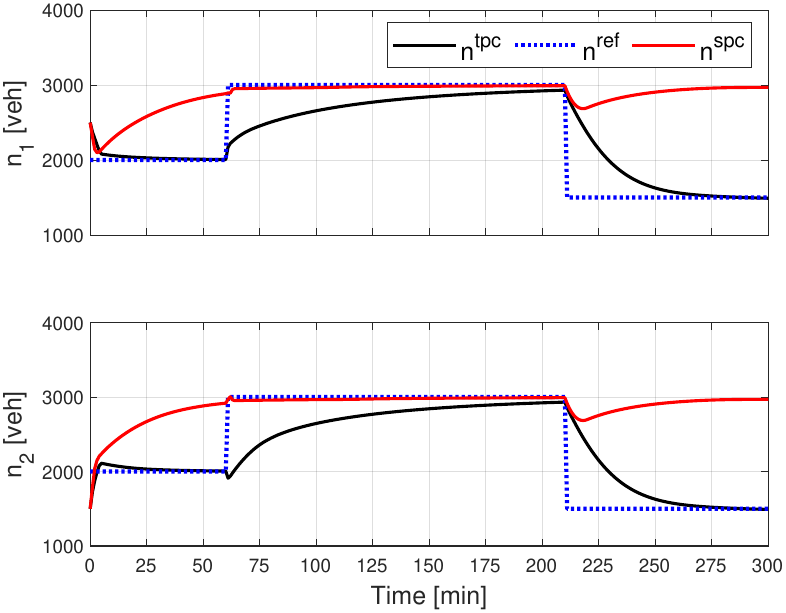}\\
  \caption{Accumulation state evolutions of Example 1}\label{fig:state2e2}
\end{figure}

\begin{figure}[htbp]
  \centering
  \includegraphics[width=3.3in]{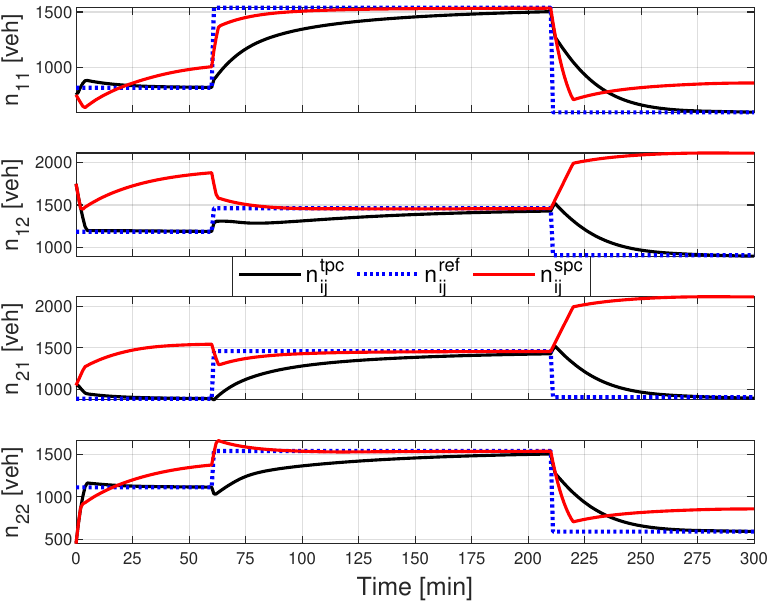}\\
  \caption{OD-specific state evolutions of Example 1}\label{fig:state4e2}
\end{figure}

A summary of two important performance indices: 1) minimizing the total time spent (TTS) and 2) maximizing the cumulative trip completion (CTC), achieved by the two controllers is given by \autoref{tb:pce1}.
The performance comparison is depicted in \autoref{fig:ttsctce2}.
Note that the proposed TPC scheme achieves a 20.01\% reduction in TTS compared with the SPC scheme. When performing congestion offset, the SPC scheme restricts the transfer flows between the two regions to keep regulating the accumulation state to the critical point, which makes the network denser and leads to unnecessary travel delays. On the contrary, by defining a reference signal that is bounded by more practical control targets, the proposed TPC scheme not only achieves a significant improvement in minimizing the TTS compared with the SPC strategy, but also can outperform the latter in facilitating the CTC.

\begin{table}[!htb]
  \centering
  \caption{Performance in TTS (veh$\cdot$s) and CTC (veh) of Example 1}
  \label{tb:pce1}
  \begin{tabular}{|c!{\vrule width 1.5pt}c!{\vrule width 1.5pt}c|}
    % \toprule
    \hline\hline
    Controller     & TTS ($\times 1e7$ veh$\cdot$s) & CTC ($\times 1e4$ veh) \\
    % \midrule
    \hline
    TPC      & 8.311 (-20.01\%) & 9.698 (+3.15\%) \\
    \hline
    SPC      & 10.391 (-) & 9.402 (-) \\
    % \bottomrule
    \hline\hline
  \end{tabular}
\end{table}

\begin{figure}[htbp]
  \centering
  \includegraphics[width=3.3in]{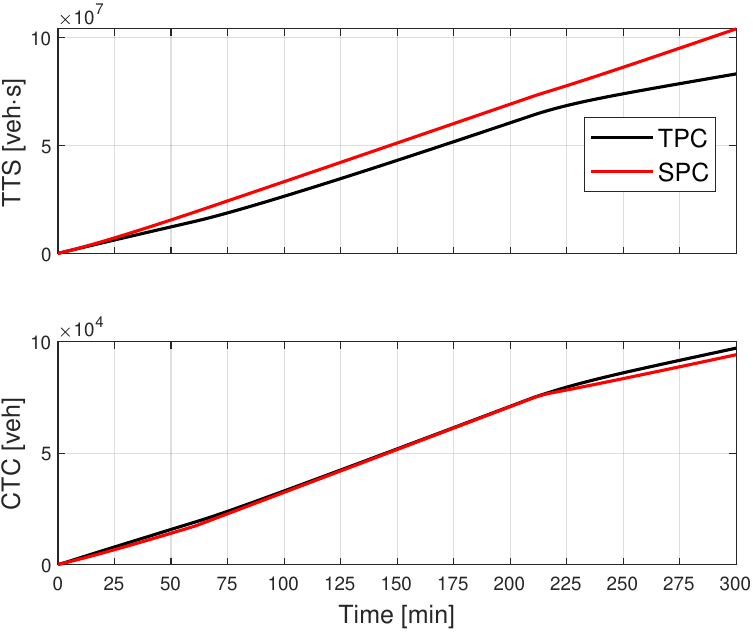}\\
  \caption{Performances in TTS and CTC of Example 1}\label{fig:ttsctce2}
\end{figure}

\subsection{Example 2: Robust trajectory tracking and trip completion maximization}

Note that the accumulation reference trajectories in Example 1 are a priori given constant accumulation points.
Now in this case study, we examine the proposed ADP tracking perimeter controller for time-varying accumulation reference trajectories.
We consider a practical urban traffic operation target, that is, to maximize the trip completion under a time-varying demand pattern.
The desired trajectory reference is given by
\begin{equation}\label{eq:ex22traj}
    \begin{split}
        \dot{n}_{d,1}(t) = \hat{q}_1(t) - G_1(n_{d,1}(t)) + \frac{n_{d,21}(t)}{n_{d,2}(t)}G_2(n_{d,2}(t)) \\
        \dot{n}_{d,2}(t) = \hat{q}_2(t) - G_2(n_{d,2}(t)) + \frac{n_{d,12}(t)}{n_{d,1}(t)}G_1(n_{d,1}(t))
    \end{split}
\end{equation}
where the initial reference signal is $n_{d,1}(0)=0,\ n_{d,2}(0)=0$.
The initial OD-specific accumulations of the original MFD system are $n_{11}(0) = 120$ (veh), $n_{12}(0) = 280$ (veh), $n_{21}(0) = 385$ (veh), $n_{22}(0) = 165$ (veh).
The nominal travel demand pattern desired or estimated by the traffic manager is shown in \autoref{fig:demand4e22}.
Based on the desired demand pattern, the traffic manager devises the trajectory reference given by \eqref{eq:ex22traj}.
However, the actual demand pattern can vary daily and is subject to external disturbances.
Thus, the actual demand pattern differs from the desired one, as shown in \autoref{fig:demand4e22_ns}, mimicking a one-hour peak period followed by a low-demand period.

\begin{figure}[htbp]
  \centering
  \includegraphics[width=3.3in]{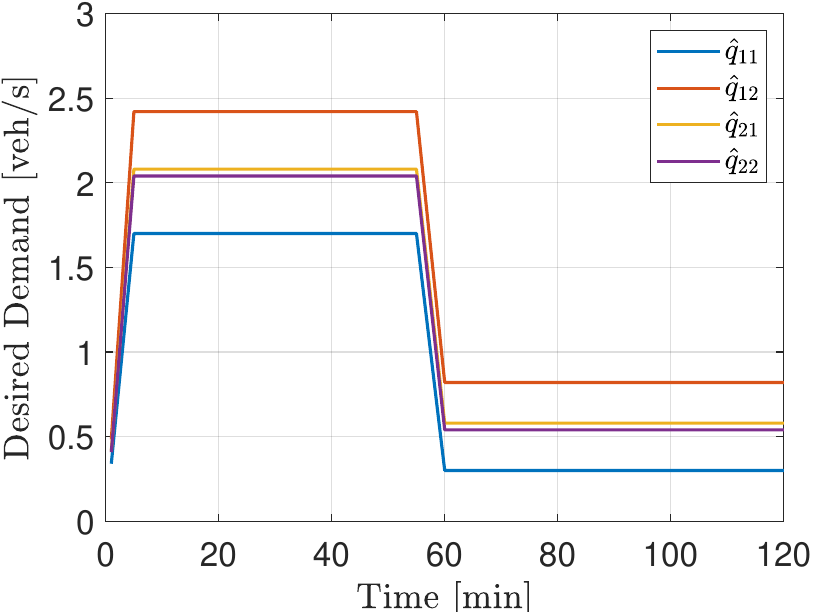}\\
  \caption{The nominal demand pattern}\label{fig:demand4e22}
\end{figure}

\begin{figure}[htbp]
  \centering
  \includegraphics[width=3.3in]{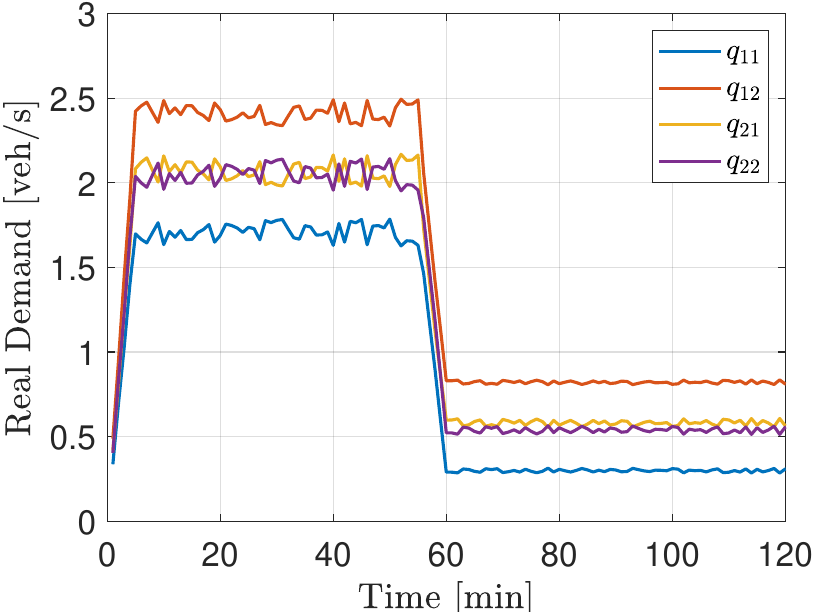}\\
  \caption{The actual noisy demand pattern}\label{fig:demand4e22_ns}
\end{figure}

The accumulation state evolution results are presented in \autoref{fig:state4e22}.
We can observe that the accumulation state can be successfully regulated to the desired reference trajectory, and the trip completion of the system is thus maximized.
The perimeter control inputs over time are shown in \autoref{fig:controle22}.
As the actual initial state is larger than the initial reference signal, the adaptive TPC metering rate first drops to around 60\% in order to track the reference trajectory.
Some chattering behavior is observed in the control input evolution. This is because the adaptive TPC controller is trying to suppress the fluctuation effect of demand disturbance on the traffic state so as to keep track of the stable reference trajectory.
The perimeter control inputs eventually approach their maximal values, which implies that almost no queuing occurs on the region boundaries and the optimal trip completion rate is achieved.
These results indicate that the proposed adaptive TPC approach can well learn the uncertain traffic dynamics and hence guarantee the robustness of the controller and the optimality of the MFD system.

\begin{figure}[htbp]
  \centering
  \includegraphics[width=3.3in]{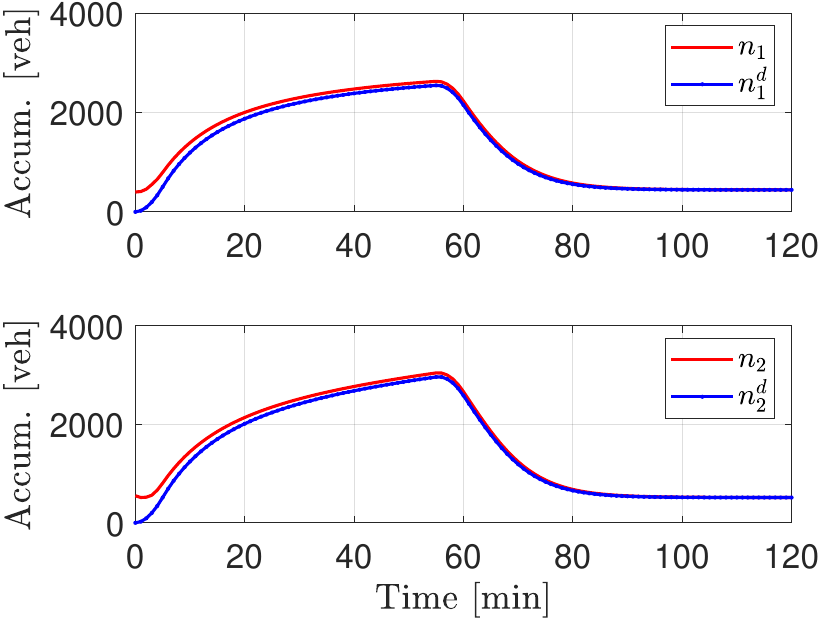}\\
  \caption{Accumulation state evolutions of Example 2}\label{fig:state4e22}
\end{figure}

\begin{figure}[htbp]
  \centering
  \includegraphics[width=3.3in]{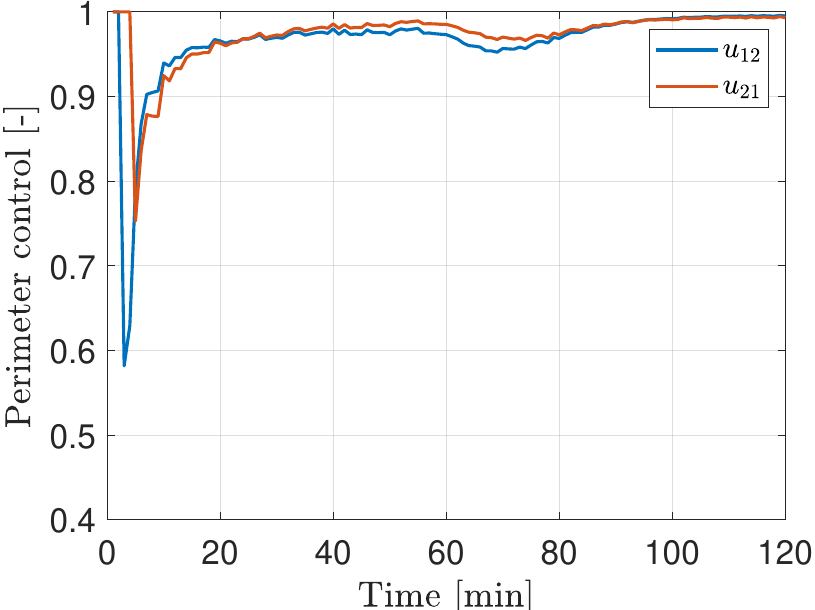}\\
  \caption{Perimeter control inputs of Example 2}\label{fig:controle22}
\end{figure}

\section{Conclusions}\label{sec:con}

This paper leverages the trajectory stability concept that can better fit the dynamic nature of traffic demand to devise an adaptive tracking perimeter control (TPC) strategy for two-region macroscopic traffic dynamics.
An adaptive dynamic programming (ADP) approach is proposed to approximate the optimal solution to the OTPCP, which requires no knowledge of the macroscopic traffic dynamics, i.e., model-free.
Compared with the traditional set-point perimeter control (SPC) scheme that tracks a pre-defined set point, the proposed ADP-based TPC scheme can well adapt to the changes in the traffic condition (e.g., time-varying travel demand) and regulate the accumulation state to a desired reference trajectory that better fits the dynamics of the demand.
Moreover, in cases of unknown traffic dynamics and demand disturbances, the adaptive TPC approach can track the optimal reference trajectory to maximize the trip completion under a nominal demand pattern, which indicates the robustness of the proposed tracking perimeter controller.
Future efforts will be dedicated to 1) defining reference trajectories from the perspective of economic benefit in traffic management, and 2) incorporating the route choice behaviors of travelers.

\balance

\bibliographystyle{IEEEtran}
\bibliography{IEEEabrv,IEEEexample}

\end{document}